\documentclass[11pt]{article}

\usepackage{amsmath}
\usepackage{amssymb}
\usepackage{amsthm}
\usepackage{a4wide}

\newtheorem{theorem}{Theorem}
\newtheorem{corollary}[theorem]{Corollary}
\newtheorem{lemma}[theorem]{Lemma}

\newtheorem{definition}[theorem]{Definition}
\newtheorem{claim}[theorem]{Claim}

\newtheorem{remark}[theorem]{Remark}

\newcommand{\poly}[1]{\mathop{\textrm{poly}}(#1)}

\newcommand{\bra}[1]{\{#1\}}

\newcommand{\setcond}[2]{\left\{#1\: \middle|\: #2\right\}}
\newcommand{\transpose}[1]{#1^{\mathrm{T}}}
\newcommand{\rank}[1]{\mathop{\mathrm{rank}}(#1)}

\newcommand{\Prob}{\mathop{\textrm{Prob}}}

\newcommand{\EXP}{\mathrm{EXP}}

\newcommand{\AC}{\mathrm{AC}}

\newcommand{\field}{\mathbb{F}}

\newcommand{\naturals}{\mathbb{N}}

\newcommand{\union}{\cup}
\newcommand{\intersect}{\cap}

\newcommand{\p}{\mathcal{P}}
\newcommand{\mono}[2]{\mathrm{Mon}_{#1}(#2)}
\newcommand{\fx}{\field\langle X\rangle}
\newcommand{\corank}[1]{\mathop{\mathrm{corank}}(#1)}
\newcommand{\mc}[1]{\mathcal{#1}}
\newcommand{\mydim}[0]{\mathrm{dim}}

\newcommand{\F}{\mathbb{F}}

\newcommand{\tH}{\tilde{H}}
\newcommand{\tA}{\tilde{A}}
\newcommand{\T}[1]{\T{#1}}

\newcommand{\x}{\overline{x}}
\newcommand{\sd}{\mathrm{SizeDepth}}

\long\def\symbolfootnote[#1]#2{\begingroup%
\def\thefootnote{\fnsymbol{footnote}}\footnote[#1]{#2}\endgroup}

\title{Circuit Lower Bounds, Help Functions, and the Remote Point
  Problem}

\author{V.~Arvind and Srikanth Srinivasan\\
Institute of Mathematical Sciences\\ C.I.T Campus,Chennai  600 113,
India\\
\tt{\{arvind,srikanth\}@imsc.res.in} 
}

\begin{document}

\maketitle

\begin{abstract}
  We investigate the power of Algebraic Branching Programs (ABPs)
  augmented with help polynomials, and constant-depth Boolean circuits
  augmented with help functions. We relate the problem of proving
  explicit lower bounds in both these models to the Remote Point
  Problem (introduced in \cite{APY}). More precisely, proving lower
  bounds for ABPs with help polynomials is related to the Remote Point
  Problem w.r.t.\ the rank metric, and for constant-depth circuits
  with help functions it is related to the Remote Point Problem
  w.r.t.\ the Hamming metric. For algebraic branching programs with
  help polynomials with some degree restrictions we show exponential
  size lower bounds for explicit polynomials.
\end{abstract}

\section{Introduction}

The goal of circuit complexity, which is central to computational
complexity, is proving lower bounds for explicit functions. The area
has made several advances in the last three decades mainly for
restricted circuit models. Some of the major results relating to
circuit size lower bounds are the following: Exponential size lower
bounds for constant-depth Boolean circuits \cite{Ha,Sm,FSS} and for
monotone Boolean circuits \cite{AB,Ra} computing certain explicit Boolean
functions; in the arithmetic circuit complexity setting, exponential
size lower bounds for monotone arithmetic circuits \cite{JS} computing
certain explicit polynomials, and exponential size lower bounds for
explicit polynomials in the case of \emph{noncommutative} algebraic
branching programs \cite{N}. More recently, \cite{R} has shown
superpolynomial lower bounds for multilinear arithmetic circuits. We
can say that these restricted models of computation have been
sufficiently well understood to show the nontrivial explicit lower bounds.

However, most of the central problems in the area continue to remain
open. For example, we do not know how to prove superlinear size lower
bounds for logarithmic depth Boolean circuits. We do not have
superpolynomial size lower bounds for depth-3 arithmetic circuits over
rationals.

The aim of this paper is to explore circuit complexity by augmenting
the power of some of these restricted models by allowing \emph{help
  functions} (in the arithmetic circuit case, \emph{help
  polynomials}). In this paper we consider two specific problems.

\begin{enumerate}
\item Proving size lower bounds for constant depth Boolean circuits
  augmented with help functions. More precisely, given any set
  $\{h_1,h_2,\cdots,h_m\}$ of help Boolean functions where
\[  
h_i:\{0,1\}^n\longrightarrow \{0,1\}, 
\]
and $m$ is (quasi)polynomial in $n$, we want
to find an explicit Boolean function $f:\{0,1\}^n\longrightarrow
\{0,1\}$ that requires superpolynomial size constant depth circuits
$C$ that takes as input $x_1,\cdots,x_n,h_1,\cdots,h_m$. The function
$f$ should be explicit in the sense that it is computable in
$2^{n^{O(1)}}$ time.

\item Proving size lower bounds for \emph{noncommutative} algebraic
  branching programs augmented with help polynomials. More precisely,
  given any set $\{h_1,h_2,\cdots,h_m\}$ of help polynomials in the
  noncommuting variables $\{x_1,x_2,\cdots,x_n\}$ over a field $\F$,
  we consider algebraic branching programs whose edges are labeled by
  $\F$-linear combinations of the $h_i$. The problem is to prove
  superpolynomial lower bounds for some explicit polynomial in
  $x_1,\cdots,x_n$ over $\F$.
\end{enumerate}

We formally define explicit Boolean functions and explicit
polynomials. 

We say that a family of Boolean functions $\{f_n\}_{n>0}$, where
$f_n:\{0,1\}^n\longrightarrow \{0,1\}$ for each $n$, is
\emph{explicit} if there is a uniform $2^{n^{O(1)}}$ time algorithm
that takes $x\in\{0,1\}^n$ as input and computes $f_n(x)$. 

We say that a family of \emph{multilinear} polynomials $\{P_n\}_{n>0}$
where $P_n(\x)\in\F[x_1,\cdots,x_n]$ is \emph{explicit} if there is a
uniform $2^{n^{O(1)}}$ time algorithm that takes as input $(m,0^n)$ for
a multilinear monomial $m$ (on indeterminates $x_1,x_2,\ldots,x_n$)
and outputs the coefficient of $m$ in the polynomial $P_n$.


\subsection*{Contributions of this paper}

For constant-depth circuits and noncommutative ABPs, augmented with
help functions/polynomials respectively, proving lower bounds appears
to be nontrivial. 

\begin{enumerate}

\item We show that both the above lower bound problems are related to
  the \emph{Remote Point Problem} studied by Alon et al \cite{APY}.
  For constant-depth circuits we show a connection to the \emph{Remote
    Point Problem} in the Hamming metric studied in \cite{APY}. For
  noncommutative ABPs the problem is connected to the \emph{Remote
    Point Problem} in the \emph{rank metric} which is defined as the
  rank distance between matrices. 

\item We also study the Remote Point Problem in the Rank metric, and
  we build on ideas from Alon et al's work (for the Hamming metric
  version) in \cite{APY} to give a deterministic polynomial-time
  algorithm for certain parameters. However, these parameters are not
  sufficient to prove lower bounds for ABPs augmented with help
  polynomials. Similarly, the parameters achieved by the algorithm in
  \cite{APY} for the Hamming metric are not sufficient to prove
  explicit lower bounds for constant-depth circuits with help
  functions.

\item On the positive side, when the degrees of the help polynomials
  are somewhat restricted, using our solution to the Remote Point
  Problem w.r.t.\ the rank metric, we show exponential size lower
  bounds for noncommutative ABPs computing certain explicit
  polynomials (e.g.\ Theorem~\ref{thm_high_deg_bound}).
\end{enumerate}

\section{Constant Depth Circuits with Help Functions}
\label{section_rpp}

In this section, we address the problem of proving lower bounds for
constant depth circuits of polynomial size that have access to help
functions $\{h_1,h_2,\cdots,h_m\}$ at the input level. Our goal is to
show how the problem is related to the Remote Point Problem w.r.t. the
Hamming metric. 

Notice that we can consider the circuit inputs $x_1,x_2,\cdots,x_n$ to
be included in the set of help functions. Thus, we can assume that we
consider constant depth circuits with input $h_1,h_2,\cdots,h_m$ and
our goal is to prove superpolynomial lower bounds for such circuits.
Notice that we cannot predetermine a hard Boolean function as the hard
function chosen will depend on $h_1,h_2,\cdots,h_m$. 

It is well known that constant depth circuits can be well approximated
by polylogarithmic degree polynomials, for different notions of
approximation. We state the results of Tarui \cite{Ta} (also see
\cite{BRS}) in the form that we require. In what follows, the field we
work in will be $\field_2$, but our results can be stated over any
constant sized field, and over the rationals.

A polynomial $p(x_1,x_2,\cdots,x_n,r_1,\cdots,r_k)$ is called a
\emph{probabilistic polynomial} if it has as input the standard input
bits $x_1,x_2,\ldots,x_n$ and, in addition, random input bits
$r_1,r_2,\ldots,r_k$. We say that the polynomial $p$ represents a
Boolean function $f:\{0,1\}^n\longrightarrow \{0,1\}$ with error
$\epsilon$ if
\begin{align*}
\Prob[p(x_1,\cdots,x_n,r_1,\cdots,r_k)=f(x_1,\cdots,&x_n)]\\\geq
&1-\epsilon,
\end{align*}
where the probability is over random choices of bits
$r_j$.

\begin{theorem}{\rm\cite{Ta,BRS}}
  There is a probabilistic polynomial
  $p(x_1,x_2,\cdots,x_n,r_1,\cdots,r_k)$ of degree
  $O(\log(1/\epsilon)\log^2 n)$ with $O(\log(1/\epsilon)\log^2 n)$
  random bits that represents $OR(x_1,\cdots,x_n)$ with error
  $\epsilon$. Furthermore, $AND(x_1,\cdots,x_n)$ can be similarly
  represented.
\end{theorem}

Building on the above, the following well-known theorem is shown in
\cite{Ta,BRS}.

\begin{theorem}\label{tarui-thm}{\rm\cite{Ta,BRS}}
	Every function $f$ computed by a boolean circuit of depth $d$ and
	size $s$ is represented by a probabilistic polynomial
	$p(x_1,x_2,\cdots,x_n,r_1,\cdots,r_k)$ of degree
  $O(\log(1/\epsilon)\log^2 n)^d$ that represents $f(x_1,\cdots,x_n)$
  with error $s\epsilon$.
	\symbolfootnote[2]{Tarui's construction yields a probabilistic polynomial
  $q$ with \emph{integer} coefficients. We can obtain the desired
  polynomial $p$ over $\field_2$ from $q$ by reducing the coefficients
  modulo $2$.}

\end{theorem}

Now, consider Boolean functions computed by constant-depth circuits
with help functions. More precisely, let $H=\{h_1,h_2,\cdots,h_m\}$
denote a set of Boolean help functions
$h_i:\{0,1\}^n\longrightarrow\{0,1\}$. For $s,d\in\naturals$, we
define $\sd_H(s,d)$ to be 
the set of Boolean functions $f:\{0,1\}^n\longrightarrow\{0,1\}$
such that there is a depth $d$ circuit $C$ of size at most $s$
such that
\[
f(\x)= C(h_1(\x),h_2(\x),\cdots,h_m(\x)),
\]
where $\x$ denotes the $n$-tuple $(x_1,x_2,\cdots,x_n)$. The lower
bound problem is to construct, for each fixed $d$, and for any given
set of help functions $H$ and $s\in\naturals$, an explicit Boolean
function $g$ such that $g$ is \emph{not} in $\sd_H(s,d)$.

We do not have a solution to this problem. However, we show that this
lower bound problem is connected to the Remote Point Problem (RPP) 
introduced by Alon et al \cite{APY}. An interesting deterministic
algorithm for RPP is presented in \cite{APY}. A deterministic
algorithm with somewhat stronger parameters would solve our lower
bound question. We now explain this connection.

\noindent\textbf{The Remote Point Problem (RPP) \cite{APY}.}~~Given
a $k$-dimensional subspace $V\subseteq \F_2^N$ the problem is to find
a vector $v\in \F_2^N$ such that the Hamming distance $d(u,v)\geq r$
for every $u\in V$ if it exists,. We will call an efficient algorithm
that does this an $(N,k,r)$-solution to the problem. 

The challenge is to give an efficient deterministic algorithm for
RPP. A randomized algorithm that simply picks $v$ at random would be a
good solution with high probability (for most parameters $k$ and $r$
of interest). Alon et al in \cite{APY} give an $(N,k,r)$ solution for
$r=O\left({\frac{N\log k}{k}}\right)$, where their deterministic algorithm runs
in time polynomial in $N$. We now state and prove the connection
between RPP and our lower bound question.

\begin{theorem}
  Let $N=2^n$. For any constant $d\in\naturals$, and any constants
  $c_0>c_1>c_2>0$ such that $c_0 > (c_1+2c_2)d + c_2$, if the Remote
  Point Problem with parameters $(N,k,r)$ -- for $k=2^{(\log n)^{c_0}}$
	and $r={\frac{N}{2^{(\log n)^{c_1}}}}$ -- can be solved in time
	$2^{n^{O(1)}}$, then, for any given set of help functions $H$ such
	that $|H| = 2^{(\log n)^{c_2}}$ and $s = cn^c$, there is an explicit
	Boolean function that does not belong to $\sd_H(s,d)$ for large
	enough $n$ (depending on $c$).
\end{theorem}
\begin{proof}
  The proof is an easy application of Theorem~\ref{tarui-thm}. Let $H
  = \bra{h_1,h_2,\ldots,h_m}$.  Consider a circuit $C$ corresponding
  to the class $\sd_H(s,d)$.  To wit, the function it computes is
  $C(h_1(\x),h_2(\x),\cdots,h_m(\x))$, where $C$ is depth-$d$,
  unbounded fanin and of size $cm^c$. Now, for $\x$ picked uniformly
  at random from $\{0,1\}^n$ suppose the probability distribution of
  $(h_1(\x),h_2(\x),\cdots,h_m(\x))$ on the set $\{0,1\}^m$ is
  $\mu$. By Theorem~\ref{tarui-thm} there is a probabilistic
  polynomial $p(y_1,y_2,\cdots,y_m,r_1,r_2,\cdots,r_t)$ of degree
  $O(\log(1/\epsilon)\log^2 m)^d$ that represents
	$C(y_1,y_2,\cdots,y_m)$ with error $cm^c\epsilon$. By a standard
	averaging argument it follows that we can fix the random bits
	$r_1,r_2,\cdots,r_t$ to get
\begin{align*}
\Prob_\mu[p(y_1,y_2,\cdots,&y_m,r_1,r_2,\cdots,r_t)=\\
&C(y_1,y_2,\cdots,y_m)]\geq 1-cm^c\epsilon,
\end{align*}
where $(y_1,y_2,\cdots,y_m)$ is picked according to distribution $\mu$.
But that is equivalent to 
\begin{align}
\Prob[p(&h_1(\x),\cdots,h_m(\x),r_1,r_2,\cdots,r_t)=\notag\\
&C(h_1(\x),h_2(\x),\cdots,h_m(\x))]\geq 1-cm^c\epsilon,\notag\\
\label{help-eqn}
\end{align}
where $\x$ is picked uniformly at random from $\{0,1\}^n$.

Choose $c_0'< c_0-c_2$ and $c_1' > c_1 (> c_2)$ such that $c_0' =
(c_1'+2c_2)d$.  Let $\epsilon={\frac{1}{2^{(\log n)^{c_1'}}}}$. Then
the degree of $p$ above is $O(\log n)^{c_0'}$. We will consider
Boolean functions on $n$ bits as vectors in $\F_2^N$.  Let $V$ be the
subspace in $\F_2^N$ spanned by all monomials (i.e, products of help
functions) of degree at most $O(\log n)^{c_0'}$. Then the dimension
$k$ of $V$ is $m^{O(\log n)^{c_0'}} < 2^{(\log n)^{c_0}}$. By 
Inequality (\ref{help-eqn}), it follows that finding a vector
$v\in\F_2^N$ that is $r$-far from $V$ for $r = \frac{N}{2^{(\log
n)^{c_1}}} > cNm^c\epsilon$ in time $2^{n^{O(1)}}$ would give us an
explicit Boolean function that is not in $\sd_H(s,d)$.
\end{proof}

\begin{remark}
  We recall a nice related result of Jin-Yi Cai: He has shown in
  \cite{Ca} an exponential lower bound for the size of constant-depth
  circuits that computes $m$ specific parities in the presence of
  (any) $m-1$ help functions, where $m\leq n^{1/5}$. His proof is
  essentially based on Smolensky's dimension argument
  \cite{Sm}. However, in our setting where we allow for polynomially
	many help functions Smolensky's argument \cite{Sm} does not work.
\end{remark}

We now state an interesting connection between explicit lower bounds
against $\sd_H(n^c,d)$ and lower bounds against the polynomial time
many-one closure of $\AC^0$.  The proof proceeds by a simple
diagonalization argument. For any complexity class $\mathcal{C}$, let
$\mathcal{R}^p_m(\mathcal{C})$ denote the polynomial-time many-one
closure of $\mathcal{C}$, i.e, the class of languages that can be
reduced in polynomial time to a language in $\mathcal{C}$.

\begin{theorem}\label{thm_exp_ac0}
  Suppose, for every fixed $d\in\naturals$, there is a $2^{n^{O(1)}}$
  time algorithm $\mathcal{A}$ that takes as input a set of help
	functions $H = \setcond{h_i:\bra{0,1}^n\rightarrow\bra{0,1}}{i\in
	[m]}$ where $m \leq n^{\log n}$ \symbolfootnote[3]{Here, $\log n$
	can be replaced by any function $f:\naturals\rightarrow\naturals$
	such that $f(n)$ is $2^{n^{O(1)}}$-time computable, $f(n) =
	\omega(1)$, and $f(n)\leq n^{O(1)}$.} (where each $h_i$ is given by
	its truth-table), and $\mathcal{A}$ outputs the truth-table of a
	Boolean function $g:\bra{0,1}^n\rightarrow\bra{0,1}$ such that for
	any $c>0$, $g\notin \sd_H(n^c,d)$ for almost all $n$. Then
	$\EXP\nsubseteq \mathcal{R}^p_m(\AC^0)$.
\end{theorem}
\begin{proof}
	For any $d\in\naturals$, let $\AC^0_d$ denote the class of languages
	that are accepted by polynomial-sized circuit families of polynomial
	size and depth $d$.

	Note that to prove that $\EXP\nsubseteq\mathcal{R}^p_m(\AC^0)$, it
	suffices to prove that $\EXP\nsubseteq\mathcal{R}^p_m(\AC^0_d)$ for
	each fixed $d\in\naturals$, since $\EXP$ contains problems that are
	complete for it under polynomial-time many-one reductions. We will
	now describe, for any fixed $d\in\naturals$, an $\EXP$ machine that
	accepts a language $L_d\notin \mathcal{R}^p_m(\AC^0_d)$.

	We proceed by diagonalization. Let $R_1,R_2,R_3,\ldots$ be any
	standard enumeration of all polynomial-time many-one reductions
	such that each reduction appears infinitely often in the list. Fix
	$n\in\naturals$ and let $m = \max_{y\in\bra{0,1}^n} |R_n(y)|$. On an
	input $x\in\bra{0,1}^n$, the $\EXP$ machine does the following: for
	each $y\in \bra{0,1}^n$, it runs $R_n$ for $n^{\log n}$ time and
	computes $R_n(y)$ (if $R_n$ does not halt in time $n^{\log n}$, the
	machine outputs $0$ and halts). It can thus produce the truth tables of
	functions $h_i:\bra{0,1}^n\rightarrow \bra{0,1}$ ($i\in[m]$) such that
	for each $y\in\bra{0,1}^n$, $h_i(y)$ is the $i$th bit of $R_n(y)$ if
	$|R_n(y)|\geq i$ and $0$ otherwise. Now, by assumption, in time
	$2^{n^{O(1)}}$, the $\EXP$ machine can compute the truth table of a
	function $g_n:\bra{0,1}^n\rightarrow\bra{0,1}$ such that, for any
	$c>0$, $g_n\notin \sd_{\bra{h_1,\ldots,h_m}}(n^c,d)$ for large enough
	$n$.  Having computed $g_n$, the $\EXP$ machine just outputs $g_n(x)$.

	It is clear, by a standard argument, that $L_d$ cannot be
	polynomial-time many-one reduced to any language in $\AC^0_d$.
\end{proof}

\section{Noncommutative Algebraic Branching Programs}
\label{section_noncomm_abps_intro}

Let $X=\{x_1,x_2,\cdots,x_n\}$ be a set of $n$ noncommuting variables,
and $\fx$ denote the non-commutative ring of polynomials over $X$ with
coefficients from the field $\field$. For $f\in\fx$, let $d(f)$
denote the degree of $f$. Let $\mono{d}{X}$ be the set of degree $d$
monomials over $X$. For a polynomial $f$ and a monomial $m$ over $X$,
let $f(m)$ denote the coefficient of $m$ in $f$. A nonempty subset
$H\subseteq\fx$ is \emph{homogeneous} if there is a $d\in\naturals$
such that all the polynomials in $H$ are homogeneous of degree $d$.

Let $G = (V,E)$ be a directed acyclic graph.  For $u,v\in V$, let
$\p_{u,v}$ be the set of paths from $u$ to $v$, where a path in
$\p_{u,v}$ is a tuple of the form $((u_0 , u_1 ), (u_1 , u_2 ),\ldots
, (u_{l-1} ,u_l ))$ where $u_0 = u$ and $u_l = v$.

\begin{definition}
  Let $X=\bra{x_1,x_2,\ldots,x_n}$ and $Y=\bra{y_1,y_2,\ldots,y_m}$ be
  disjoint variable sets. Let
  $H=\bra{h_1,h_2,\ldots,h_m}\subseteq\fx$. An \emph{Algebraic
    Branching Program (ABP) with help polynomials $H$} is a layered
  directed acyclic graph $A$ with a source $s$ and a sink $t$. Every
  edge $e$ of $A$ is labeled by a linear form $L(e)$ in variables
  $X\union Y$. If $L(e) = \sum_i\alpha_ix_i+\sum_j\beta_jy_j$, the
  \emph{polynomial $L'(e)$ associated with edge $e$} is obtained by
  substituting $h_j$ for $y_j$, $1\leq j\leq m$, in $L(e)$. I.e.\
  $L'(e) = \sum_i\alpha_ix_i + \sum_j\beta_jh_j$. The \emph{size} of
  $A$ is the number of vertices in $A$.
\end{definition}

Given a path $\gamma = (e_1,e_2,\ldots,e_l)$ in $A$, define the
polynomial $f_\gamma= L'(e_1)\cdot L'(e_2)\cdot\ldots\cdot L'(e_t)$ (note that
the order of multiplication is important). For vertices $u$ and $v$ of
$A$, we define the polynomial
$f_{u,v}=\sum_{\gamma\in\p_{u,v}}f_\gamma$. The ABP $A$ computes
the polynomial $f_{s,t}$.

Suppose $L(e) = \sum_i\alpha_ix_i + \sum_j\beta_jy_j$. We say that the
edge $e$ is \emph{homogeneously labeled} if all the polynomials in the
set $\setcond{x_i}{\alpha_i\neq 0}\union\setcond{h_j}{\beta_j\neq 0}$
are homogeneous and of the same degree $d(e)$. If the above set is
empty, we let $d(e)=0$. Now, suppose all edges of an ABP $A$ are
homogeneously labeled; then, for a path $\gamma = (e_1,e_2,\ldots,e_t)$ in
$A$ let $d(\gamma)=\sum_{i=1}^td(e_i)$. The ABP $A$ with help
polynomials $H$ is \emph{homogeneous} if:
\begin{itemize}
\item all the edges in $A$ are homogeneously labeled,
\item For all $u,v$ in $A$ and $\gamma_1,\gamma_2\in\p_{u,v}$,
  $d(\gamma_1) = d(\gamma_2)$.
\end{itemize}

For a homogeneous ABP $A$ with help polynomials and any pair of
vertices $u,v$ in $A$, the polynomial computed from $u$ to $v$ is
homogeneous.

%
In the absence of help polynomials, this gives the standard Algebraic
Branching Programs as defined in, e.g.\ Nisan \cite{N}. Nisan \cite{N}
has shown explicit lower bounds, e.g.\ for the Permanent and
Determinant, for this model of computation. Our aim is to prove lower
bounds for ABPs with help polynomials.

We show that any ABP with arbitrary help polynomials computing a
homogeneous polynomial can be transformed into an equivalent
\emph{homogeneous} ABP with homogeneous help polynomials with only a
small increase in size. Thus, it suffices to prove lower bounds
against homogeneous ABPs with help polynomials.  Fix the help
polynomial set $H\subseteq\fx$. Let $m=|H|$ and $d(H)=\max_{h\in
H}d(h)$. Also, fix some $d\in\naturals$.

Given $f\in \fx$ and $i\in\naturals$, let $f^{(i)}$ denote the $i$th
homogeneous part of $h$. For $2\leq i\leq d$, let $\tH_i =
\setcond{h^{(i)}\in\fx}{h\in H}$; let $\tH = \bigcup_{2\leq i}\tH_i$.
Let $\tilde{m}_i$ denote $|\tH_i|$ for each $i$, and let
$\tilde{m}$ denote $|\tH| = \sum_i\tilde{m}_i$. We show the following
homogenization theorem.
\begin{theorem}\label{thm_homogen}
Given any ABP $A$ using the help polynomials $H$ computing a
homogeneous polynomial $f$ of degree $d\geq 1$, there is a homogeneous
ABP $\tA$ using the help polynomials $\tH$ that computes the same
polynomial as $A$, where the size of $\tA$ is at most $S(d+1)$, where
$S$ denotes the size of $A$.
\end{theorem}
\begin{proof}
The following construction is fairly standard.  Let $s$ and $t$ be the
designated source and sink, respectively, of the ABP. We will use the
notation of Section \ref{section_noncomm_abps_intro}.

We now define $\tA$. $\tA$ will use the variables $X\union \tilde{Y}$, where
$\tilde{Y} =
\setcond{y_i^{(j)}}{1\leq i\leq m, 2\leq j\leq d(h_i)}$.
The vertices of $\tA$ are tuples $(u,i)$, where
$u$ is a vertex of $A$ and $i\in\naturals$ is a number between $0$ and $d$.
The source of $\tA$ will be $(s,0)$ and the sink $(t,d)$.  
We will define the set of edges of $\tA$ in two stages. We will first
construct an ABP on the set of vertices of $\tA$ which will include
edges with weights from $\field$ (i.e, edges $e$ such that $L(e)$ is a
non-zero degree $0$ polynomial), and we will then show how to remove
these edges from the ABP. Consider any edge $e$ in the ABP $A$; let
the label $L(e)$ of $e$ be $\sum_{i=1}^{n}\alpha_ix_i +
\sum_{j=1}^{m}\beta_jy_j$ and
$0\leq k\leq d$, define the linear form $L(e)_k$ -- which captures the
$k$th homogeneous part of $L'(e)$, the polynomial computed by edge $e$
-- as follows:
\begin{itemize}
\item If $k = 0$, define $L(e)_k$ to be the field element
$\sum_{j=1}^{m}\beta_jh_j^{(0)}$. 
\item If $k = 1$, define $L(e)_k$ to be $\sum_{i=1}^{n}\alpha_ix_i +
\sum_{j=1}^{m}\beta_jh_j^{(1)}$.
\item If $k > 1$, define $L(e)_k$ to be
$\sum_{j=1}^{m}\beta_jy_j^{(k)}$
\end{itemize}

Fix any vertex $(v,k)$ of $\tA$. Let $\bra{u_1,u_2,\ldots,u_l}$ be the
predecessors of $v$ in $A$ and let $e_i$ denote the edge $(u_i,v)$.
Then, it is easy to see that 
\[f_{s,v}^{(k)} = \sum_{i=1}^{l}\sum_{j=0}^{k}
f_{s,u_i}^{(j)}L'(e_i)_{(k-j)}\]

Hence, we define edges $e_{i,j}$ in $\tA$ from vertices $(u_i,j)$
to $(v,k)$ with label $L(e_{i,j}) = L(e_i)_{k-j}$. (Note that the
label $L(e_{i,k})$ is just a field element. We will change this
presently.) This concludes the first stage. Note that, since we only
add edges from $(u,i)$ to $(v,j)$ when $(u,v)$ is an edge in $A$,
the graph of $\tA$ is acyclic. Also note that an edge $e$ is labeled by
a field element if and only if it connects vertices of the form
$(u,k)$ and $(v,k)$, for some $u$, $v$, and $k$. Finally, it is easily
seen from the definition of $\tA$ that the polynomial computed from
$(s,0)$ to $(u,i)$ is the polynomial $f_{s,u}^{(i)}$ for any $s,u,$
and $i$.

In the second stage, we will get rid of those edges in $\tA$ such that
$L(e)\in\field$. We do this in two passes. Fix some topological
ordering of the vertices of $\tA$, and order the edges
$(\tilde{u},\tilde{v})$ of $\tA$ lexicographically. As long as there
is an edge $e = (\tilde{u},\tilde{v})$ of $\tA$ such that $\tilde{v}$
is \emph{not} the designated sink $(t,d)$ and $L(e)\in\field$, we let
$e$ be the least such edge and do the following: we remove the edge
$e$, and for each edge $e' = (\tilde{v},\tilde{w})$ of $\tA$ going out
of $v$, we change the label of the edge $e'' = (\tilde{u}, \tilde{w})$
to $L(e'') + L(e)\cdot L(e')$ (if no such edge $e''$ exists, we add
this edge to the ABP and give it the label $L(e)\cdot L(e')$). It
should be clear that the homogeneity of the ABP is preserved. After at
most $O((sd)^2)$ many such modifications, all edges in $\tA$ that are
labeled by field elements are of the form $(\tilde{u},(t,d))$.
Moreover, by the above construction, it is clear that $\tilde{u} =
(u,d)$ for some vertex $u\neq t$ of $A$. Since $d\geq 1$, we know that
$\tilde{u}\neq (s,0)$, the designated source node. We also know that
there are no edges into $\tilde{u}$ which are labeled by a field
element. We now do the following: for each edge $e =
(\tilde{u},(t,d))$ labeled by a field element, we remove the vertex
$\tilde{u}$ and for each edge $e' = (\tilde{v},\tilde{u})$, we remove
$e'$ and change the label of $e'' = (\tilde{v},(t,d))$ to $L(e'') +
L(e')\cdot L(e)$ (if no such $e''$ exists, we add such an edge $e''$
and set its label to $L(e')\cdot L(e)$). This concludes the
construction.

It is easy to prove inductively that after every modification of
$\tA$, the polynomial computed from $(s,0)$ to $(t,d)$ remains
$f_{s,t}^{(d)}$.  Hence, the ABP $\tA$ computes exactly the polynomial
$f$ computed by $A$.  Also, by construction, the edges of $\tA$ are
all homogeneously labeled; finally, it can also be seen that given a
path $\gamma$ from vertex $(u,i)$ to vertex $(v,j)$ in $\tA$,
$d(\gamma) = j-i$: hence, the ABP is indeed homogeneous, and we are
done.  
\end{proof}

\section{Decomposition of Communication Matrices}

We now generalize the key lemma of Nisan \cite{N} that connects the
size of noncommutative ABPs for an $f\in\fx$ to the ranks of certain
communication matrices $M_k(f)$. The generalization is for
noncommutative ABPs with help polynomials, and it gives a more
complicated connection between the size of ABPs to the ranks of
certain matrices. For usual noncommutative ABPs considered in
\cite{N}, Nisan's lemma directly yields the lower bounds. In our case,
this generalization allows us to formulate the lower bound problem as
a Remote Point Problem for the rank metric.  

We will assume that the explicit polynomial for which we will be
proving lower bounds is homogeneous. Thus, by Theorem
\ref{thm_homogen} we can assume that each help polynomial in
$H=\bra{h_1,h_2,\ldots,h_m}$ is homogeneous and of degree at least
$2$.

We first fix some notation. Let $d\in\naturals$ be an even number. Let
$d(H) = \max_{h\in H}d(h)$. Also, for $2\leq i\leq d(H)$, let $H_i =
\setcond{h\in H}{d(h) = i}$.

Suppose $f\in\fx$ is homogeneous of even degree $d\geq 2$, and
$k\in\naturals$ such that $0\leq k \leq d$. We define the $n^k\times
n^{d-k}$ matrix $M_k(f)$ (as in \cite{N}): Each row is labeled by a
distinct monomial in $\mono{k}{X}$ and each column by a distinct
monomial in $\mono{d-k}{X}$. Given monomials $m_1\in\mono{k}{X}$ and
$m_2\in\mono{d-k}{X}$, the $(m_1,m_2)$th entry of $M_k(f)$ is the
coefficient of the monomial $m_1m_2$ in $f$ and is denoted by
$M_k(f)(m_1,m_2)$.

Call $M$ an \emph{$(l,m)$-matrix} if $M$ is an $n^l\times n^m$ matrix
with entries from $\field$, where the rows of $M$ are labeled by
monomials in $\mono{l}{X}$ and columns by monomials in
$\mono{m}{X}$. Suppose $0\leq l\leq k$ and $0\leq m \leq d-k$. Let
$M_1$ be an $(l,m)$-matrix and $M_2$ a $(k-l,(d-k)-m)$-matrix. We
define the $(k,d-k)$-matrix $M = M_1\otimes_{l,m}^k M_2$ as follows:
Suppose $m_1\in\mono{k}{X}$ and $m_2\in\mono{d-k}{X}$ are monomials
such that $m_1 = m_{11}m_{12}$ with $m_{11}\in\mono{k-l}{X}$ and
$m_{12}\in\mono{l}{X}$ and $m_2 = m_{21}m_{22}$ with
$m_{21}\in\mono{m}{X}$ and $m_{22}\in\mono{(d-k)-m}{X}$. Then the
$(m_1,m_2)^{th}$ entry of $M$ is defined as 
\[
M(m_1,m_2) = M_1(m_{12},m_{21})\cdot M_2(m_{11},m_{22}).
\]

Let $A$ be a homogeneous ABP with help polynomials $H$ computing a
polynomial $f$ of degree $d$. Let $u,v$ and $w$ be vertices in the ABP
$A$, and $\gamma_1\in\p_{u,v}$ and $\gamma_2\in\p_{v,w}$ be paths.  We
denote by $\gamma_1\circ\gamma_2\in\p_{u,w}$ the concatenation of
$\gamma_1$ and $\gamma_2$. 
%

Since $A$ is homogeneous, each of the polynomials $f_{u,v}$ for
vertices $u, v$ of $A$ is homogeneous. For $1\leq k\leq d/2$, define
the \emph{$k$-cut} of $A$, $C_k\subseteq V(A)\union E(A)$, as follows:
A vertex $v\in V(A)$ is in $C_k$ iff $d(f_{s,v}) = k$, and an edge $e
= (u,v)\in E(A)$ is in $C_k$ iff $d(f_{s,u}) < k$ and
$d(f_{s,v})>k$. For each $x\in C_k$, let $\p_x$ denote the set of
$s$-$t$ paths passing through $x$. Clearly, the sets
$\setcond{\p_x}{x\in C_k}$ partition $\p_{s,t}$, the set of all paths
from $s$ to $t$. Thus, we have
\begin{align}\label{eq_1}
f &~= \sum_{x\in C_k}\sum_{\gamma\in\p_x} f_\gamma\notag\\
  &~= \sum_{v\in C_k\intersect V(A)}\sum_{\gamma\in\p_v} f_\gamma ~+ 
     \sum_{e\in C_k\intersect E(A)}\sum_{\gamma\in\p_e}
     f_\gamma.\notag\\
\end{align}

We now analyze Equation \ref{eq_1}. For $v\in C_k\intersect V(A)$,
$\p_v = \setcond{\gamma_1\circ\gamma_2}{\gamma_1\in\p_{s,v},
  \gamma_2\in\p_{v,t}}$. Hence, for any $v\in C_k\intersect V(A)$:
\begin{align}\label{eq_2}
\sum_{\gamma\in\p_v} f_\gamma &= \sum_{\substack{\gamma_1\in\p_{s,v}\\
\gamma_2\in\p_{v,t}}} f_{\gamma_1\circ\gamma_2}
= \sum_{\substack{\gamma_1\in\p_{s,v}\\
\gamma_2\in\p_{v,t}}} f_{\gamma_1}\cdot f_{\gamma_2}\notag\\
&= f_{s,v}f_{v,t}.\notag\\
\end{align}

Similarly, for any edge $e=(u,v)\in C_k\intersect E(A)$, $\p_e =
\setcond{\gamma_1\circ (e)\circ \gamma_2}{\gamma_1\in\p_{s,u},
\gamma_2\in\p_{v,t}}$, where $(e)$ denotes the path containing just
the edge $e$. Thus,
\begin{align}\label{eq_3}
\sum_{\gamma\in\p_e} f_\gamma &= \sum_{\substack{\gamma_1\in\p_{s,u}\\
\gamma_2\in\p_{v,t}}} f_{\gamma_1\circ (e)\circ \gamma_2}\\
&= \sum_{\substack{\gamma_1\in\p_{s,v}\notag\\
\gamma_2\in\p_{v,t}}} f_{\gamma_1}\cdot L'(e)\cdot f_{\gamma_2}\\
&=f_{s,u}L'(e)f_{v,t}.\notag\\
\end{align}

{From} Equations \ref{eq_1}, \ref{eq_2}, and \ref{eq_3}, we get
\begin{align*}
f = \sum_{v\in C_k\intersect V(A)}f_{s,v}f_{v,t} ~+\\ 
\sum_{e=(u,v)\in C_k\intersect E(A)}f_{s,u}L'(e)f_{v,t}. 
\end{align*}
As $A$ is homogeneous of degree $d$, each polynomial in the sums above
is homogeneous of degree $d$. Hence
\begin{align}\label{eq_4}
M_k(f) = &\sum_{v\in C_k\intersect V(A)}M_k(f_{s,v}f_{v,t})\ +\notag\\
	   &\sum_{e=(u,v)\in C_k\intersect
		 E(A)}M_k(f_{s,u}L'(e)f_{v,t}).\notag\\
\end{align}

For any $v\in C_k\intersect V(A)$, $f_{s,v}$ and $f_{v,t}$ are
homogeneous degree $k$ and $d-k$ polynomials respectively. We denote
by $M_v$ the matrix $M_k(f_{s,v}f_{v,t})$.  Notice that for
$m_1\in\mono{k}{X}$ and $m_2\in\mono{d-k}{X}$, the $(m_1,m_2)^{th}$
entry of the matrix $M_v=M_k(f_{s,v}f_{v,t})$ is
$f_{s,v}(m_1)f_{v,t}(m_2)$. Thus, $M_v$ is an outer product of two
column vectors and is hence a matrix of rank at most $1$. Therefore,
the first summation in Equation~\ref{eq_4} is a matrix of rank at most
$|C_k\intersect V(A)|$.  


For $e=(u,v)\in C_k\intersect E(A)$, we know that $d(f_{s,u})<k$ and
$d(f_{s,v})>k$ and thus, $d(e)\geq 2$. Hence, $L'(e) = \sum_{h\in
  H_{d(e)}}\beta_{e,h} h$, for $\beta_{e,h}\in\field$. Therefore,
expanding the second summation in Equation \ref{eq_4}, we get
\begin{align}
\sum_{\substack{e=(u,v)\in\\ C_k\intersect
E(A)}}&M_k(f_{s,u}L'(e)f_{v,t}) =\notag\\
&\sum_{\substack{e=(u,v)\in\\ C_k\intersect E(A)}}\sum_{h\in H_{d(e)}}
\beta_{e,h}M_k(f_{s,u}\cdot h\cdot f_{v,t})\notag\\
\label{eq_5}
\end{align}

Consider a term of the form $M_k(f_{s,u}hf_{v,t})$. For the rest of
the proof let $d(w)$ denote $d(f_{s,w})$, for any vertex $w$ of
$A$. Given monomials $m_{11}\in\mono{d(u)}{X}$,
$m_{12}\in\mono{k-d(u)}{X}$, $m_{21}\in\mono{d(h)-(k-d(u))}{X}$, and
$m_{22}\in\mono{d-d(v)}{X}$, the entry
$M_k(f_{s,u}hf_{v,t})(m_{11}m_{12},m_{21}m_{22}) = 
h(m_{12}m_{21})f_{s,u}(m_{11})f_{v,t}(m_{22})$,
since all polynomials involved are homogeneous. Hence, the matrix
$M_k(f_{s,u}hf_{v,t})$ is precisely
$M_{k-d(u)}(h)\otimes_{k-d(u),d(h)-(k-d(u))}^kM_e$, where
$M_e(m_{11},m_{22}) = f_{s,u}(m_{11})f_{v,t}(m_{22})$, for
$m_{11}\in\mono{d(u)}{X}, m_{22}\in\mono{d-d(v)}{X}$. Clearly, $M_e$
is a matrix of rank at most $1$, for any $e\in C_k\intersect E(A)$ and
$h\in H_{d(e)}$. Continuing with the above calculation, we get
\begin{align*}
&\sum_{\substack{e=(u,v)\in\\ C_k\intersect
E(A)}}M_k(f_{s,u}L'(e)f_{v,t})\\
&=\sum_{\substack{e=(u,v)\in\\ C_k\intersect E(A)}}\sum_{h\in H_{d(e)}}
\beta_{e,h}M_{l_e}(h)\otimes_{l_e,m_e}^kM_e\\
&= \sum_{h\in H}\sum_{i=d_1(h)}^{d_2(h)} M_i(h)\otimes_{i,d(h)-i}^k
\cdot \sum_{\substack{e=(u,v)\in C_k:\\
d(e)=d(h)\\d(u) = k-i}}\beta_{e,h}M_e,
\end{align*}
where $d_1(h) = \max\bra{1,d(h)-(d-k)}$, $d_2(h) =
\min\bra{d(h)-1,k}$, $l_e=k-d(u)$, and $m_e = d(h)-(k-d(u))$.

Plugging the above observations into Equation \ref{eq_4}, we have
\begin{align*}
&M_k(f) = \underbrace{\left(\sum_{v\in C_k\intersect
V(A)}M_v\right)}_{M'} +\\ 
&\sum_{h\in H}\sum_{i=d_1(h)}^{d_2(h)} M_i(h)\otimes_{i,d(h)-i}^k
\underbrace{\left(\sum_{\substack{e=(u,v)\in C_k:\\
	 d(e) = d(h)\\d(u) = k-i}}\beta_{e,h}M_e\right)}_{M'_{i,h}}
\end{align*}

Notice that $M'$ above has rank at most $|V(A)|$, and $M'_{i,h}$ has
rank at most $|E(A)|\leq |V(A)|^2$ for any $h\in H$ and $d_1(h)\leq i\leq
d_2(h)$. Hence, we have proved the following result:

\begin{theorem}
\label{thm_decomp}
Let $A$ be a homogeneous ABP of size $S$ computing a (homogeneous)
polynomial $f$ of degree $d$ using the help polynomials $H$. Then, for
any $k\in \bra{0,1,\ldots,d}$, we can write $M_k(f)$ as:
\begin{align*}
&M_k(f) = M' +
&\sum_{h\in
  H}\sum_{i=d_1(h)}^{d_2(h)}M_i(h)\otimes_{i,d(h)-i}^k
M'_{i,h},
\end{align*}
where $d_1(h) = \max\bra{1,d(h)-(d-k)}$ and $d_2(h) =
\min\bra{d(h)-1,k}$ such that $\rank{M'}\leq S$ and
$\rank{M'_{i,h}}\leq S^2$ for each $h\in H$, and
$i\in\bra{\max\bra{1,d(h)-(d-k)}\ldots,\min\bra{d(h)-1,k}}$.
\end{theorem}

\section{Remote Point Problem for the rank metric}
\label{section_rmp}

We now introduce an algorithmic problem that will help us prove lower
bounds on the sizes of ABPs computing explicit polynomials using a
(given) set of help polynomials $H$. This problem is actually the
\emph{Remote Point Problem for matrices in the rank metric} that we
denote RMP. This problem is analogous to the Remote Point
Problem (RPP), which we discussed in Section~\ref{section_rpp}.

Given two matrices $P,Q\in\field^{a\times b}$, the \emph{Rank
  distance} between $P$ and $Q$ is defined to be $\rank{P-Q}$. It is
known that this defines a metric, known as the \emph{rank metric} on
the set of all $a\times b$ matrices over $\field$.

\noindent\textbf{The RMP problem.}~~ Given as input a set of
$N\times N$ matrices $P_1,P_2,\ldots,P_k$ over a field $\field$ and
$r\in\naturals$, the problem is to compute an $N\times N$ matrix $P$
such that for any matrix $P'=\sum_{i=1}^k \alpha_iP_i$ in the subspace
generated by $P_1,P_2,\ldots,P_k$, the rank distance between $P$ and
$P'$ is at least $r$.

In the problem $N$ is taken as the input size, and $k$ and $r$ are
usually functions of $N$. We say that the RMP problem has an 
$(N,k,r)$-solution over $\field$ if there is a \emph{deterministic}
algorithm that runs in time polynomial in $N$ and computes a matrix
$P$ that is at rank distance at least $r$ from the subspace generated by the
$P_1,P_2,\ldots,P_k$.

\begin{remark}\label{rem1}
How does a solution to RMP give us an explicit noncommutative
polynomial $f$ for which we can show lower bounds for the sizes of
noncommutative ABPs with help polynomials? We now explain the
connection.

Let $A$ be a homogeneous ABP of size $S$ computing a polynomial
$f$ of degree $d$. Let $d_1(h)$ denote $\max\bra{1,d(h)-d/2}$ and
$d_2(h)$ denote $\min\bra{d/2,d(h)-1}$. For $a,b,p,q\in\naturals$ such
that $p\in [n^a]$ and $q\in [n^b]$, let $E_{a,b}^{p,q}$ be the $n^a\times
n^b$ elementary matrix with $1$ as $(p,q)$th entry, and $0$
elsewhere. The matrices $\setcond{E_{a,b}^{p,q}}{p\in [n^a],q\in
  [n^b]}$ span all matrices in $\field^{n^a\times n^b}$.
By Theorem \ref{thm_decomp}
\[
M_{d/2}(f) = M' + \sum_{h\in H}
\sum_{i=d_1(h)}^{d_2(h)}M_i(h)\otimes_{i,d(h)-i}^{d/2} M'_{i,h},
\]
where $\rank{M'}\leq S$. For $h\in H$ and
$i\in\bra{d_1(h),\ldots,d_2(h)}$, the matrix $M'_{i,h}$ is an
$n^{d/2-i}\times n^{d/2-d(h)+i}$ dimension matrix. We can write
$M'_{i,h}$ as a linear combination of the elementary matrices
in $\{E_{d/2-i,d/2-d(h)+i}^{p,q}\mid p\in [n^{d/2-i}], q\in
  [n^{d/2-d(h)+i}]\}$.

Let $\mathcal{A}$ be the set of matrices of the form
$M_i(h)\otimes_{i,d(h)-i}^{d/2}E_{d/2-i,d/2-d(h)+i}^{p,q}$, where
$h\in H$, $i\in\bra{d_1(h),\ldots,d_2(h)}$, and $p\in
[n^{d/2-i}]$, $q\in [n^{d/2-d(h)+i}]$. Each matrix in $\mathcal{A}$ is
an $n^{d/2}\times n^{d/2}$ matrix, with its rows and columns labeled
by monomials in $\mono{d/2}{X}$. Every matrix of the form
$M_i(h)\otimes_{i,d(h)-i}^{d/2} M'_{i,h}$ is a linear combination of
matrices in $\mathcal{A}$. Crucially, note that $\mathcal{A}$ depends
only on the set of help polynomials and the parameter $d$, and it
\emph{does not depend} on the ABP $A$.

By substitution for $M'_{i,h}$ we obtain the following expression for
$M_{d/2}(f)$ in terms of linear combination of matrices in
$\mathcal{A}$.
\begin{equation*}
M_{d/2}(f) = M' + \sum_{M\in \mathcal{A}}\alpha_M M,
\end{equation*}
where $\alpha_M\in\field$. Since, $M'$ has rank at most $S$, it
implies that $M_{d/2}(f)$ is at rank distance at most $S$ from the
subspace generated by the matrices in $\mathcal{A}$. Thus, if we can
compute a matrix $\hat{M}$ in deterministic time polynomial in $n^d$
that has rank distance $S=2^{O(n)}$ from the subspace generated by
$\mathcal{A}$ we would obtain an explicit homogeneous degree $d$
polynomial $f$ with lower bound $2^{\Omega(n)}$ by setting
$\hat{M}=M_{d/2}(f)$. This is the approach that we will take 
for proving lower bounds.
\end{remark}

We present the following simple algorithm, which suffices for our
lower bound application. 

\begin{theorem}
\label{thm_rmp_simple}
For any $k$, the RMP has an $(N,k,\lfloor N/k+1\rfloor)$-solution over
any field $\field$ such that field operations in $\field$ and Gaussian
elimination over $\field$ can be performed in polynomial time.
\end{theorem}
\begin{proof}
We assume that $k<N$; otherwise the problem is trivial. Let $r$ denote
$\lfloor N/k+1 \rfloor$. Choose the first $r$ column vectors in each
of the matrices $P_1,P_2,\ldots,P_k$. Let
$v_1,v_2,\ldots,v_{rk}\in\field^N$ be
these vectors in some order. As $rk\leq N - r$, using Gaussian
elimination, we can efficiently choose
$v_{rk+1},v_{rk+2},\ldots,v_{r(k+1)}\in\field^N$ with the following
property: for every $i\in [k+1]$, $v_{rk+i}$ is linearly independent
of $v_1,v_2,\ldots,v_{rk+(i-1)}$. Let $P$ be any matrix that has
$v_{rk+1},v_{rk+2},\ldots,v_{r(k+1)}$ as its first $r$ columns. It is
not too difficult to see that given any matrix $P'$ in the subspace
generated by $P_1,P_2,\ldots,P_k$, the first $r$ columns of $P-P'$
remain independent, i.e $\rank{P-P'}\geq r$.
\end{proof}

\begin{remark}
	The Remote Point Problem is fascinating as an algorithmic question.
	In \cite{APY} Alon et al provide a nontrivial algorithm for RPP in
	the Hamming metric (over $\F_2$). We use similar methods to provide
	an improved solution to RMP for small prime fields. The result is
	proved in Section \ref{section_rmp_improv}. Unfortunately, the
	improvement in parameters over the trivial solution above is not
	enough to translate into an appreciably better lower bound.
\end{remark}

\section{Lower bounds for ABPs with Help Polynomials}

In this section, we prove some lower bounds for ABPs computing some
explicit polynomials using a set of given help polynomials $H$. Here,
`explicit' means that the coefficients of the polynomial can be
written down in time polynomial in the number of coefficients of
the input (the help polynomials $H$) and the output (the hard to
compute polynomial).

Throughout this section, $\field$ will be a field over which field
operations and Gaussian elimination can be performed efficiently. Let
the set of help polynomials be $H=\bra{h_1,h_2,\ldots,h_m}$; let
$d(H)=\max_{h\in H}d(h)$. 

We will first consider the case of homogeneous ABPs using the help
polynomials $H$; $H$ is, in this case, assumed to be a set of homogeneous
polynomials. We will then derive a lower bound for general ABPs and a
general set of help polynomials using Theorem \ref{thm_homogen}.

\subsection{The homogeneous case}
\label{subsec_homogen}

Let $H$ be a set of homogeneous polynomials in this section. Our aim
is to produce, for any degree $d\in\naturals$, an explicit homogeneous
polynomial $F_d$ of degree $d$ that cannot be computed by homogeneous
ABPs. To avoid some trivialities, we will assume that $d$ is even. 

We first observe that, to compute homogeneous polynomials of degree
$d$, a homogeneous ABP cannot meaningfully use help polynomials of
degree greater than $d$:

\begin{lemma}
\label{lemma_high_deg_no_use}
Let $A$ be a homogeneous ABP using the help polynomials $H$ to compute
a polynomial $f$ of degree $d$. Then, there is a homogeneous ABP $A'$,
of size at most the size of $A$, such that $A'$ computes $f$ and
furthermore, for every edge $e\in E(A')$, $d(e)\leq d$.
\end{lemma}
\begin{proof}
Simply take $A$ and throw away all edges $e\in E(A)$ such that $d(e)>
d$; call the resulting homogeneous ABP $A'$. Since $A$ is homogeneous,
no path from source to sink in $A$ can contain an edge $e$ that was
removed above. Hence, the polynomial computed remains the same.
\end{proof}

Hence, to prove a lower bound for an explicit homogeneous polynomial
of degree $d$, it suffices to prove a lower bound on the sizes of
ABPs computing this polynomial using the help polynomials
$H_{\leq d} = \setcond{h\in H}{d(h)\leq d}$. As above, let $d(H_{\leq d}) =
\max_{h\in H_{\leq d}}d(h)$.

We begin with a simple explicit lower bound. Call a homogeneous
polynomial $F\in\fx$ of degree $d$ \emph{$d$-full-rank} if
$\rank{M_{d/2}(F)} = n^{d/2}$. Full-rank polynomials are easily
constructed; here is a simple example of one:
$F(X) = \sum_{m\in\mono{d/2}{X}}m\cdot m$. It follows easily from
Nisan's result \cite{N} that, without any help polynomials,
homogeneous ABPs computing any $d$-full-rank polynomial are of size at
least $n^{d/2}$.

\begin{theorem}
\label{thm_low_deg_bound}
Assume that $d(H_{\leq d})\leq d(1-\epsilon)$, for a fixed constant $\epsilon>0$
and let $F\in\fx$ be a $d$-full-rank polynomial. Then, any homogeneous ABP $A$
computing $F$ has size at least $\left( n^{\frac{\epsilon
d}{4}}/\sqrt{2md}\right)$.
\end{theorem}
\begin{proof}
Consider a homogeneous ABP $A$ computing $F$ using the help
polynomials $H$. By the above lemma, we may assume that $A$ uses only
the polynomials $H_{\leq d}$. Let $S$ denote the size of $A$. For any
$h\in H_{\leq d}$, let $d_1(h)$ denote $\max\bra{1,d(h)-d/2}$ and
$d_2(h)$ denote $\min\bra{d/2,d(h)-1}$. By Theorem \ref{thm_decomp}, we
know that
\[M_{d/2}(F) = M' + \sum_{h\in H_{\leq
d}}\sum_{i=d_1(h)}^{d_2(h)}M_i(h)\otimes_{i,d(h)-i}^{d/2} M'_{i,h}\]
where $\rank{M'}\leq S$ and $\rank{M'_{i,h}}\leq S^2$, for each $h\in
H_{\leq d}$ and $i\in\bra{d_1(h),\ldots,d_2(h)}$. For any $h$ and any $i$ such that $0\leq i\leq
d(h)$, $\rank{M_i(h)}\leq \min\bra{n^{i},n^{d(h)-i}}$, which is at
most $n^{d(h)/2}\leq n^{d(H_{\leq d})/2}$. By our assumption on
$d(H_{\leq d})$, we
see that $\rank{M_i(h)}\leq n^{(1-\epsilon)d/2}$. By the definition of
$\otimes_{i,d(h)-i}^{d/2}$, this implies that
$\rank{M_i(h)\otimes_{i,d(h)-i}^{d/2} M'_{i,h}}\leq
\rank{M_i(h)}\cdot\rank{M'_{i,h}}$, which is at most
$n^{(1-\epsilon)d/2}S^2$. Thus, we see that
\begin{align*}
\rank{M_{d/2}(F)}& \leq S +
\sum_{h\in H_{\leq d}}\sum_{i=d_1(h)}^{d_2(h)}n^{(1-\epsilon)d/2}S^2\\
&\leq S + |H_{\leq d}|dn^{(1-\epsilon)d/2}S^2\\
&\leq  2mdS^2n^{(1-\epsilon)d/2}
\end{align*}
As $F$ is $d$-full-rank, this implies that
\begin{align*}
2mdS^2n^{(1-\epsilon)d/2}&\geq n^{d/2}\\
\therefore\quad S &\geq \frac{n^{\frac{\epsilon d}{4}}}{\sqrt{2md}}
\end{align*}
\end{proof}

The above theorem tells us that as long as the help polynomials are
not too many in number ($m=n^{o(d)}$ will do), and of degree at most
$(1-\epsilon)d$, then any full rank polynomial remains hard to compute
for ABPs with these help polynomials.

We now consider the case when $d(H_{\leq d})$ can be as large as $d$.
In this case, we are unable to come up with an unconditional explicit
lower bound. A strong solution to the RMP introduced in Section
\ref{section_rmp} would give us such a bound. However, with the
suboptimal solution of Theorem \ref{thm_rmp_simple}, we are able to
come up with explicit lower bounds in a special case.  Let $\delta(H)$
denote $\min_{h\in H}d(h)$. By assuming some lower bounds on
$\delta(H)$, we are able to compute an explicit hard function.

\begin{theorem}
\label{thm_high_deg_bound}
Assume $\delta(H)\geq (\frac{1}{2}+\epsilon)d$, for a fixed
constant $\epsilon > 0$. Then, there exists an explicit homogeneous
polynomial $F\in\fx$ of degree $d$ such that any homogeneous ABP $A$
computing $F$ using the help polynomials $H$ has size at least
$\lfloor n^{\frac{\epsilon d}{2}}/2md\rfloor$.
\end{theorem}

\begin{proof}
  Let $A$ be a homogeneous ABP $A$ of size $S$ computing a polynomial
  $f$ of degree $d$. Let $d_1(h)$ denote $\max\bra{1,d(h)-d/2}$ and
  $d_2(h)$ denote $\min\bra{d/2,d(h)-1}$. As explained in
  Remark~\ref{rem1}, let $E_{a,b}^{p,q}$ denote the $n^a\times
	n^b$-sized elementary matrix with $1$ in the $(p,q)$th entry and $0$s
	elsewhere. The matrices $\setcond{E_{a,b}^{p,q}}{p\in [n^a],q\in
	[n^b]}$ span all $n^a\times n^b$ matrices.

By Theorem \ref{thm_decomp}
\[
M_{d/2}(f) = M' + \sum_{h\in H_{\leq
d}}\sum_{i=d_1(h)}^{d_2(h)}M_i(h)\otimes_{i,d(h)-i}^{d/2} M'_{i,h}
\]
where $\rank{M'}\leq S$. As explained in Remark~\ref{rem1}, 
$M'_{i,h}$ is an $n^{d/2-i}\times n^{d/2-d(h)+i}$ dimension matrix and
is in the span of $\bra{E_{d/2-i,d/2-d(h)+i}^{p,q}}$, where $p\in
  [n^{d/2-i}], q\in [n^{d/2-d(h)+i}]$.

Let $\mathcal{A}$ denote the set of $n^{d/2}\times n^{d/2}$ matrices
of the form
$M_i(h)\otimes_{i,d(h)-i}^{d/2}E_{d/2-i,d/2-d(h)+i}^{p,q}$, where
$h\in H_{\leq d}$, $i\in\bra{d_1(h),\ldots,d_2(h)}$, and $p\in
[n^{d/2-i}]$, $q\in [n^{d/2-d(h)+i}]$. Then we obtain
\begin{equation}
\label{eq_6}
M_{d/2}(f) = M' + \sum_{M\in \mathcal{A}}\alpha_M M,
\end{equation}
where $\alpha_M\in\field$. Since $M'$ is a matrix of rank at most
$S$, this implies that $M$ is at rank distance at most $S$ from the
subspace generated by the matrices in $\mathcal{A}$.

Let $k=|\mathcal{A}|$. For each $h\in H$ and
$i\in\bra{d_1(h),\ldots,d_2(h)}$, we have added precisely $n^{d-d(h)}$
many matrices of the form $M_i(h)\otimes_{i,d(h)-i}^{d/2}E$, where $E$
is an elementary matrix of dimension $n^{d/2-i}\times n^{d/2-d(h)+i}$.
Since $d(h)\geq d(\frac{1}{2}+\epsilon)$ for each $h\in H_{\leq
d}\subseteq H$, this implies that $k\leq
mdn^{\frac{d}{2}(1-\epsilon)}$. Let $N$ denote $n^{d/2}$;
$\mathcal{A}$ consists of $k\leq mdN^{1-\epsilon}$ $N\times N$
matrices. By Theorem \ref{thm_rmp_simple}, we can, in time $\poly{N}$,
come up with an $N\times N$ matrix $M_0$ that is at rank distance at
least $\lfloor\frac{N}{k+1}\rfloor$ from the subspace generated by the
matrices in $\mathcal{A}$. We label the rows and columns of $M_0$ by
monomials from $\mono{d/2}{X}$, in the same way as the matrices in
$\mathcal{A}$ are labeled. Using $M_0$, we define the homogeneous
degree $d$ polynomial $F\in\fx$ to be the unique polynomial such that
$M_{d/2}(F) = M_0$; that is, given any monomial $m\in\mono{d}{X}$ such
that $m = m_1\cdot m_2$ for $m_1,m_2\in\mono{d/2}{X}$, $F(m)$ is
defined to be $M_0(m_1,m_2)$.

Let $A$ be a homogeneous ABP of size $S$ computing $F$ using the help
polynomials $H$. Then, by Equation \ref{eq_6} we have
\[M_{d/2}(F) = M' + \sum_{M\in\mathcal{A}}\alpha_M M\]
where $\alpha_M\in\field$, and $\rank{M'}\leq S$. Since $M_{d/2}(F)$ is
$M_0$, which is at rank distance at least $\lfloor N/(k+1)\rfloor$
from the subspace generate by $\mathcal{A}$, we see that
$S\geq \rank{M'}\geq \lfloor N/(k+1)\rfloor$. This implies that,
\[S\geq \left\lfloor\frac{N}{mdN^{1-\epsilon}+1} \right\rfloor\geq
\left\lfloor\frac{N^{\epsilon}}{2md}\right\rfloor = 
\left\lfloor\frac{n^{\frac{\epsilon d}{2}}}{2md}\right\rfloor\]
\end{proof}

\begin{remark}
	The rather unnatural condition on $\delta(H)$ above can be removed
	with better solutions to the RMP problem. Specifically, one can show
	along the above lines that if the RMP has an $(N,k,N/k^{1/2 -
	\epsilon})$-solution for $k=N^{2\delta}$, then for any $H$, there
	is an explicit polynomial that cannot be computed by any ABP $A$
	using $H$ of size at most $n^{\Omega(\epsilon d)}/(md)^{O(1)}$.
	Here, $\epsilon$ and $\delta$ are arbitrary constants in $(0,1)$.
\end{remark}

\subsection{The inhomogeneous case}

Let $\tH$ denote the set of all homogeneous
parts of degree at least $2$ obtained from polynomials in $H$, i.e
$\tH = \setcond{h_j^{(i)}}{j\in [m], 2\leq i \leq d(h_j)}$. For $2\leq
i\leq d(H)$, let $\tH_i = \setcond{h\in\tH}{d(h) = i}$. Note that $\tH
= \bigcup_{2\leq i \leq d(H)}\tH_i$. 

As in the previous subsection, we construct explicit hard polynomials
for even $d\in\naturals$. Let $\tH_{\leq d}$ denote $\bigcup_{2\leq
i\leq d}\tH_i$ if $d\leq d(H)$, and $\tH$ otherwise.

\begin{corollary}
\label{coro_gen_low_deg}
Assume $d(\tH_{\leq d})\leq d(1-\epsilon)$, for a fixed constant
$\epsilon>0$. Then, there is an explicit homogeneous polynomial $F$ of
degree $d$ such that any ABP that computes $F$ using the help
polynomials $H$ has size at least $\frac{n^{\frac{\epsilon
d}{4}}}{\sqrt{2m}d(d+1)}$.
\end{corollary}
\begin{proof}
Let $F$ be a $d$-full-rank polynomial, as defined in Section
\ref{subsec_homogen}.  Consider any ABP $A$ computing $F$ using $H$.
By Theorem \ref{thm_homogen}, there exists a homogeneous ABP $\tA$
computing $F$ using $\tH$, where the size of $\tA$ is at most
$S(d+1)$. By Lemma \ref{lemma_high_deg_no_use}, we may assume that
$\tA$ uses only the help polynomials in $\tH_{\leq d}$. Since
$|\tH_{\leq d}|\leq md$, Theorem \ref{thm_low_deg_bound} tells us that
$S(d+1)\geq n^{\frac{\epsilon d}{4}}/\sqrt{2md^2}$, which implies the
result.
\end{proof}

\begin{corollary}
\label{coro_gen_high_deg}
Let $\delta(\tH) = \min_{h\in\tH} d(h)$, and assume $\delta(\tH)\geq
(\frac{1}{2}+\epsilon)d$ for a fixed constant $\epsilon > 0$. Then,
there exists an explicit homogeneous polynomial $F\in\fx$ of degree
$d$ such that any ABP $A$ computing $F$ using the help polynomials $H$
has size at least $\frac{1}{d+1}\left\lfloor \frac{n^{\frac{\epsilon
d}{2}}}{2md^2}\right\rfloor$.
\end{corollary}
\begin{proof}
By Theorem \ref{thm_homogen}, given any ABP $A$ of size $S$ computing
a homogeneous polynomial of degree $d$, there is a homogeneous ABP
$\tA$ of size at most $S(d+1)$ that computes the same polynomial as
$A$ using the help polynomials $\tH$. By Lemma
\ref{lemma_high_deg_no_use}, we may assume that $\tA$ only uses the
help polynomials $\tH_{\leq d}$. Now, let $F$ be the explicit
polynomial from Theorem \ref{thm_high_deg_bound}, with $\tH_{\leq d}$
taking on the role of $H$ in the statement of the theorem; since
$|\tH_{\leq d}|\leq md$, Theorem
\ref{thm_high_deg_bound} tells us that $S(d+1)\geq \lfloor
n^{\frac{\epsilon d}{2}}/2md^2\rfloor$, which implies the result.
\end{proof}

\section{A better solution to the RMP}
\label{section_rmp_improv}

Following the approach of Alon et al \cite{APY}, who provide a
nontrivial algorithm for RPP in the Hamming metric (over $\F_2$), we
improve on the parameters of Theorem \ref{thm_rmp_simple} for the RMP
over small prime fields. It is interesting to note that in our
solution we get similar parameters as \cite{APY}. As mentioned
earlier, the improvement in parameters over the simple solution of
Theorem \ref{thm_rmp_simple} is too little to give us a much better
lower bound. 

Throughout this section, $\field$ will denote a constant-sized field.
The main result is stated below.

\begin{theorem}
\label{thm_rmp_improv}
For any fixed constant $c>0$, the RMP has an $(N,\ell N,r)$-solution
over any constant-sized field $\field$ and for any $\ell,r > 0$ such that
$\ell\cdot r< c\log N$.
\end{theorem}

In proving the above theorem, we will follow the algorithm of
\cite{APY}. We need the following lemma, implicit in \cite{APY}:

\begin{lemma}
\label{lemma_union_subspaces}
Fix any field $\field$ such that Gaussian elimination over $\field$
can be performed in polynomial time.  There is a $\poly{M,m,|\field|}$
time algorithm for the following problem: Given subspaces
$V_1,V_2,\ldots,V_m$ of $\field^M$ such that $\sum_{i=1}^m|V_i|<
|\field|^M$, find a point $u\in\field^M$ such that $u\notin
\bigcup_iV_i$.
\end{lemma}
\begin{proof}
	The algorithm will fix the coordinates of $u$ one by one.
Assuming that the values $u_1,u_2,\ldots,u_i$ have been fixed for
$0\leq i\leq n$, let $U_i = \setcond{w\in \field^M}{w_j =
u_j\textrm{ for } 1\leq j\leq i}$. The algorithm will fix the
coordinates of $u$, ensuring that the following is true: For each
$i$ such that $1\leq i\leq M$, $\sum_{j=1}^m|V_j\intersect
U_i|<|U_i| = |\field|^{M-i}$. Note that, since $U_0$ is just $\field^M$, the inequality
is satisfied at $i=0$ by the assumption on the size of the subspaces
$V_1,V_2,\ldots,V_m$; also note that the inequality is satisfied at
$i=M$ if and only if $u\notin \bigcup_iV_i$. 

Assuming $u_1,u_2,\ldots,u_i$ have been fixed for $i<M$, we define,
for every $\alpha\in\field$, the set $U_{i,\alpha} = \setcond{w\in
U_i}{w_{i+1} = \alpha}$. Clearly, the sets $\bra{U_{i,\alpha}}_\alpha$
partition $U_i$. Hence, we see that
$\sum_{j=1}^m|V_j\intersect U_i| =
\sum_{\alpha\in\field}\sum_{j=1}^m|V_j\intersect U_{i,\alpha}|$ and
thus, there is some $\alpha\in\field$ such that
$\sum_{j=1}^m|V_j\intersect U_{i,\alpha}|<\frac{|U_i|}{|\field|} =
|\field^{M-i-1}|$.

Here is the algorithm:
\begin{itemize}
\item While $u_1,u_2,\ldots,u_i$ have been determined for $i<M$, do
the following:
\begin{itemize}
\item As mentioned above, the following invariant is maintained:
$\sum_{j=1}^k|V_j\intersect U_i|<|U_i| = |\field|^{M-i}$.
\item Find $\alpha\in\field$ such that $\sum_{j=1}^k|V_j\intersect
U_{i,\alpha}|<\frac{|U_i|}{|\field|} = |\field^{M-i-1}|$. By the
reasoning in the paragraph above, such an $\alpha$ exists and surely,
it can be found in $\poly{M,k,|\field|}$ time using Gaussian
elimination.
\item Set $u_{i+1}$ to $\alpha$.
\end{itemize}
\end{itemize}

The correctness of the algorithm is clear from the reasoning above.
\end{proof}

We now briefly describe the improved algorithm for the RMP. Let
$P_1,P_2,\ldots,P_k$ be the input matrices. We denote by $L$ the
subspace of $\field^{N\times N}$ spanned by these matrices. Also,
let $B_r$ denote the matrices of rank at most $r$. The idea of the
algorithm is to ``cover'' the set $L+B_r$ by a union of subspaces
$V_1,V_2,\ldots,V_m$ such that $\sum_i|V_i|<|\field|^{N^2}$. We then use
the algorithm from Lemma \ref{lemma_union_subspaces} to find a matrix
$P$ that is not in $\bigcup_i V_i$; by the way we have picked the
subspaces, it is clear that $M$ will then be at rank distance at least
$r$ from the subspace $L$.

What follows is an important definition.
\begin{definition}
	Fix positive integers $(d_1,d_2)$. Given $\mc{T}$, a collection of
	subspaces of $\field^N$, we say that $\mc{T}$ is $(d_1,d_2)$-good
	if:
	\begin{itemize}
		\item $\mydim(U)\leq N - d_1$ for each $U\in\mc{T}$.
		\item Each $A\subseteq \field^N$ of size $d_2$ is contained in
			some $U\in\mc{T}$.
	\end{itemize}
\end{definition}

The following claim illustrates the importance of $(d_1,d_2)$-good
subspaces of $\field^N$.
\begin{claim}
	\label{claim_vector_to_matrix}
	There is an algorithm that, when given as input $\mc{T}$, a 
	$(d_1,d_2)$-good collection of subspaces of $\field^N$, produces a
	collection $\mc{S}$ of subspaces of $\field^{N\times N}$ of
	cardinality at most $|\mc{T}|$, with the following properties:
	\begin{itemize}
		\item $\mydim(V)\leq N^2 - d_1N$ for each $V\in\mc{S}$.
		\item $B_{d_2}\subseteq \bigcup_{V\in\mc{S}}V$
	\end{itemize}
	Moreover, the algorithm runs in time $\poly{|\mc{T}|,N}$.
\end{claim}
\begin{proof}
	For each $U\in\mc{T}$, let $V(U)$ denote the subspace of
	$\field^{N\times N}$ generated by all vectors of the form
	$u\transpose{v}$, where $u\in U$ and $v\in \field^N$. The collection
	$\mc{S}$ is the collection of all such vector spaces $V(U)$, for
	$U\in\mc{T}$. Clearly, the cardinality of $\mc{S}$ is bounded by
	$|\mc{T}|$.

	Note that a basis for $V(U)$ can be constructed by picking only
	$u\transpose{v}$ where $u$ and $v$ range over bases for $U$ and
	$\field^N$ respectively. This shows that $\mydim(V(U))\leq N^2-d_1N$
	and that $V(U)$ can be constructed efficiently.

	Finally, given any matrix $Q$ of rank at most $d_2$, it can be
	written as a sum of matrices $Q_1+Q_2+\ldots+Q_{d_2}$, where each
	$Q_i$ is a matrix of rank at most $1$ and hence can be written as
	$u_i\transpose{v_i}$, where $u_i,v_i\in\field^N$. Let $A =
	\bra{u_1,u_2,\ldots,u_{d_2}}$. Since $\mc{T}$ is $(d_1,d_2)$-good,
	there is some $U\in\mc{T}$ such that $A\subseteq U$. This implies
	that $u_i\transpose{v_i}\in V(U)$ for each $i\in [d_2]$. As $V(U)$
	is a subspace, it must contain their sum $Q$. This concludes the
	proof.
\end{proof}

It is easily seen that a random collection of subspaces of $\field^N$
of appropriate dimension is $(d_1,d_2)$-good for the values of
$d_1$ and $d_2$ that are of interest to us. We now assert the
existence of an explicit collection of subspaces with this property.

\begin{claim}
	\label{claim_expl_const}
	Fix any constant $c\geq 1$. For any $\ell,r\in\naturals$ such that
	$\ell\cdot r<c \log N$, there is an algorithm that runs in time
	$N^{O(c)}$ and produces an $(\ell,r)$-good collection of subspaces
	of $\field^N$. 
\end{claim}

We prove the above claim in the next section. Assuming the claim, we
can prove Theorem \ref{thm_rmp_improv}.

\begin{proof}[Proof of Theorem \ref{thm_rmp_improv}]
	We will describe an algorithm for the problem. Without loss of
	generality, assume that $c\geq 1$. Let $L$ be the input
	subspace of dimension at most $\ell N$. We would like to
	find a matrix $P$ that is at rank distance at least $r$ from $L$.

	We first use the algorithm referred to in Claim
	\ref{claim_expl_const} to construct an $(\ell+1,r)$-good collection of
	subspaces $\mc{T}$ of $\field^N$ in time $N^{O(c)}$. Clearly,
	$|\mc{T}| = N^{O(c)}$. Then, we use the algorithm of Claim
	\ref{claim_vector_to_matrix} to construct a collection of subspaces
	$\mc{S}$ of $\field^{N\times N}$ of size $N^{O(c)}$ with the
	following properties:
	\begin{itemize}
		\item $\mydim(V)\leq N^2 - (\ell + 1)N$ for each $V\in\mc{S}$.
		\item $B_{r}\subseteq \bigcup_{V\in\mc{S}}V$
	\end{itemize}

	Consider the collection of subspaces $\mc{S}' = \bra{L+V\ |\
	V\in\mc{S}}$. Clearly, $L+B_r\subseteq \bigcup_{V\in\mc{S}'}V$.
	Moreover, the dimension of each subspace in $\mc{S}'$ is at most
	$\ell N + N^2 - (\ell+1)N \leq N^2 - N$. Hence, each subspace in
	$\mc{S}'$ is of cardinality at most $|\field|^{N^2-N}$. Since
	$|\mc{S}'| = N^{O(c)}$, for large enough $N$, we have
	$\sum_{V\in\mc{S}'}|V|< |\field|^{N^2}$. Hence, using the algorithm
	described in Lemma \ref{lemma_union_subspaces}, we can, in time
	$N^{O(c)}$, find a matrix $P\notin \bigcup_{V\in\mc{S}'}V$.
	By construction, this matrix $P$ is at rank distance greater than
	$r$ from the subspace $L$. The entire algorithm runs in time
	$N^{O(c)}$.
\end{proof}

\subsection{Proof of Claim \ref{claim_expl_const}}

We give two different constructions: one for the case that $\ell\geq
r$ and the other for the case that $\ell \leq r$.

The following notation will be useful. For each $i\in [N]$, let
$e_i\in\field^N$ denote the vector that has a $1$ in coordinate $i$
and is $0$ elsewhere. For any vector $x\in\field^N$ and $S\subseteq
[N]$, we denote by $x|_S$ the vector in $\field^{|S|}$ that is the
projection of $x$ to the coordinates indexed by $S$.

\subsubsection{Case 1: $\ell\geq r$}
For each $A\subseteq\field^{2\ell}$ of cardinality $r$, let $V_A$ be the
subspace generated by $\setcond{x\in\field^N}{x|_{[2\ell]} \in A}$. It
is easily seen that $\mydim(V_A)\leq N-2\ell + r\leq N- \ell$.
Moreover, given any $A_1\subseteq\field^N$ of size $r$, $A_1\subseteq
V_{A}$ where $A$ is any subset of $\field^{2\ell}$ of size $r$
containing $\setcond{x|_{[2\ell]}}{x\in A_1}$. Hence, the collection
$\mc{T} = \setcond{V_A}{A\subseteq\field^{2\ell}, |A| = r}$ is an
$(\ell,r)$-good collection of subspaces. 

The cardinality of $\mc{T}$ is $\binom{|\field|^{2\ell}}{r}\leq
|\field|^{2\ell r} = N^{O(c)}$. Surely, $\mc{T}$ can be constructed in
time $N^{O(c)}$.

\subsubsection{Case 2: $\ell\leq r$}
Given a set $A\subseteq \field^m$ for some $m\in\naturals$, we denote
by $\rank{A}$ the size of any maximal set of linearly independent
vectors from $A$; we denote by $\corank{A}$ the value $(|A|-\rank{A})$. 

Fix a set $A\subseteq\field^m$ for some $m\in\naturals$. Given
$d,d'\in\naturals$, we say that $A$ is \emph{$d$-wise corank $d'$} if
each $B\subseteq A$ such that $|B| = d$ satisfies $\corank{B}\leq d'$;
$A$ is said to be \emph{$d$-wise linearly independent} if it is
$d$-wise corank $0$. Sets that are $d$-wise linearly independent have
been studied before: see \cite[Proposition 6.5]{ABI}, where matrices
whose columns form a $d$-wise linearly independent set of vectors are
used to construct $d$-wise independent sample spaces. The following
claim follows from this result and from the lower bound on the size of
any $d$-wise independent sample space proved in \cite[Proposition
6.4]{ABI}.

\begin{claim}[implicit in \cite{ABI}]
	\label{claim_ABI}
	Consider a set $A\subseteq\field^m$ of cardinality $t$. If $A$ is
	$d$-wise linearly independent with $d\leq 2\sqrt{t}$, then $m\geq
	\frac{d\log_{|\field|} t}{5}$, for large enough $d,t$.
\end{claim}

Using the above claim, we prove the following lower bound on the size
of sets that are $d$-wise corank $d'$ for suitable $d,d'$.

\begin{claim}
	\label{claim_lbd_dwise}
	Consider a set $A\subseteq\field^r$ of cardinality $t$. There is an
	absolute constant $c_0$ such that the following holds. Let $A$ be 
	$d$-wise corank $d'$ for positive integers $d,d'$ with $c_0d'\leq d\leq
	2\sqrt{t}$. Then, $r\geq	\frac{d\log_{|\field|} t}{12d'}$ if
	$t,d,d'$ are large enough.
\end{claim}

\begin{proof}
	Denote by $d''$ the value $\lfloor d/2d'\rfloor$. We construct a
	sequence of sets $A_0,A_1,\ldots$ as follows: $A_0$ is the set $A$;
	for any $i\geq 0$, if $A_i$ has been constructed and is $d''$-wise
	linearly independent, we stop; otherwise, there is a $B\subseteq
	A_i$ of cardinality $d''$ that is not linearly independent -- in
	this case, we set $A_{i+1} = A_i\setminus B$; we stop at $i = d'$.
	It is easy to see that the cardinality $t_i$ of $A_i$ is $t-id''$.
	It can also be checked that if $A_i$ is $d_i$-wise corank $d_i'$,
	then $A_{i+1}$, if constructed, is $(d_i-d'')$-wise corank $d_i' - 1$;
	it therefore follows that the set $S_i$, if constructed, is
	$(d-id'')$-wise corank $d'-i$, for any $i\geq 0$ -- in particular,
	$S_{d'}$ is $d/2$-wise linearly independent.

	We base our analysis on when the above process stops. Let $i_0$ be
	the largest $i$ so that $A_i$ is constructed. Its size $t_{i_0}$ is
	at least $t-d'd''\geq t - d/2\geq t/2$ for large enough $t$. If
	$i_0 = d'$, then $A_{i_0}$ is a set of size at least $t/2$ that is
	$d/2$-wise linearly independent -- by Claim \ref{claim_ABI}, we get
	$r\geq \frac{d\log_{|\field|} t}{12}$ for large enough $d,t$.
	Otherwise, $i_0<d'$ and we must have $A_{i_0}$ is $d''$-wise
	linearly independent -- in this case, by Claim \ref{claim_ABI}, we
	get $r\geq \frac{d''\log_{|\field|} t}{5}\geq \frac{d\log_{|\field|}
	t}{12d'}$ if $c_0$ is large enough. Thus, in either case, our claim
	holds.
\end{proof}

Now, we apply the above lemma with $t =
|\field|^{\lceil\frac{20}{c_0}\sqrt{c\log N}\rceil}$ and $d = c_0\lceil\sqrt{c\log
N}\rceil$. We obtain the following corollary:

\begin{corollary}
	\label{coro_lbd_dwise}
	Let $t,d$ be as defined above. For large enough $N$, given any
	$A\subseteq\field^r$ of size $t$, there is a subset $B$ of $A$ of
	cardinality $d$ such that $\corank{B}\geq\ell$. 
\end{corollary}
\begin{proof}
	Assume that $A$ is $d$-wise corank $d'$ for some $d'$. We will
	show that $d'\geq \ell$. For large enough $N$, by Claim
	\ref{claim_lbd_dwise}, we have $d'\geq
	\min\bra{\frac{d}{c_0},\frac{d\log_{|\field|}t }{12 r}}$. It remains
	to be shown that this quantity is at least $\ell$.

	Note that, since $\ell\leq r$, $\ell^2\leq \ell r\leq c\log N$.
	Hence, $\ell\leq\sqrt{c\log N}$. Thus, by the choice of $d$, we see
	that $d/c_0\geq \ell$. Moreover,
	\[
	\frac{d\log_{|\field|}t}{12r}\geq \frac{20c\log N}{12r} > \ell
	\]
	Hence, we see that $d'\geq\ell$.
\end{proof}

We now define the $(\ell,r)$-good collection of subspaces. For each
$S\subseteq [t]$ of cardinality $d$, and each $A\subseteq
\field^{d}$ of size $d-\ell$, let $V_{S,A}$ be the subspace
generated by $\setcond{x\in\field^N}{x|_S = u\text{ for some $u\in A$}}$. It
can be seen that $\mydim(V_{S,A})\leq N-d+d-\ell = N-\ell$ for each
$S,A$. 

Given any $A_1\subseteq\field^N$ of cardinality $r$, let
$P\in\field^{r\times N}$ be the matrix the rows of which are the
elements of $A_1$. Let $A_2$ denote the set of the first $t$ columns
of $P$. By Corollary \ref{coro_lbd_dwise}, there is a $B\subseteq A_2$
of size $d$ such that $\corank{B}\geq\ell$. Let $S\subseteq [t]$ index
the columns of $B$ in $P$. It can be seen that $A_1\subseteq V_{S,A'}$
for any $A'$ of size $d-\ell$ containing a set that spans
$\setcond{v|_S}{v\in A_1}$ (such an $A'$ exists since
$\corank{B}\geq\ell$). 

Thus, we can take for our collection $\mc{T}$ of $(\ell,r)$-good subspaces the
collection of all $V_{S,A}$, where $S\subseteq [t]$ with $|S|=d$, and $A\subseteq
\field^{d}$ of size $d-\ell$. The size of $\mc{T}$ is bounded
by $\binom{t}{d}\binom{|\field|^d}{d-\ell}\leq t^d|\field|^{d^2} = N^{O(c)}$, by our
choice of $d$ and $t$. Clearly, $\mc{T}$ can be constructed in time
$N^{O(c)}$.

\noindent\textbf{Acknowledgments.}~~We are grateful to Jaikumar
Radhakrishnan for discussions. We also thank the anonymous referee for
useful comments and suggestions.


\begin{thebibliography}{1}

\bibitem{ABI} Noga Alon, L\'{a}szl\'{o} Babai, Alon Itai: ``A Fast and
  Simple Randomized Parallel Algorithm for the Maximal Independent Set
  Problem'', Journal of Algorithms 7(4): 567-583 (1986).

\bibitem{AB} Noga Alon, Ravi B. Boppana: ``The monotone circuit
  complexity of Boolean functions'', Combinatorica 7(1): 1-22 (1987).

\bibitem{APY} Noga Alon, Rina Panigrahy, Sergey Yekhanin:
  ``Deterministic Approximation Algorithms for the Nearest Codeword
  Problem'', Electronic Colloquium on Computational Complexity (ECCC)
  15(065): (2008).
%
%
\bibitem{BRS} Richard Beigel, Nick Reingold, Daniel A. Spielman: ``PP
  Is Closed under Intersection'', J. Comput. Syst. Sci. 50(2): 191-202
  (1995).

\bibitem{Ca} Jin-yi Cai: ``Lower Bounds for Constant-Depth Circuits in
  the Presence of Help Bits'', Inf. Process. Lett. 36(2): 79-83
  (1990).

\bibitem{FSS} Merrick L. Furst, James B. Saxe, Michael Sipser:
  ``Parity, Circuits, and the Polynomial-Time Hierarchy'',
  Mathematical Systems Theory 17(1): 13-27 (1984).

\bibitem{Ha} Johan H\r{a}stad: ``Almost Optimal Lower Bounds for Small
  Depth Circuits'', in Randomness and Computation, Advances in
  Computing Reasearch, Vol 5, ed. S. Micali, 1989, JAI Press Inc, pp
  143-170.

\bibitem{JS} Mark Jerrum, Marc Snir: ``Some Exact Complexity Results
  for Straight-Line Computations over Semirings'', J. ACM 29(3):
  874-897 (1982).

\bibitem{N} Noam Nisan: ``Lower Bounds for Non-Commutative
  Computation'' (Extended Abstract), STOC 1991: 410-418.

\bibitem{R} Ran Raz: ``Separation of Multilinear Circuit and Formula
	Size''. Theory of Computing 2(1): 121-135 (2006).

\bibitem{Ra} Alexander Razborov: ``Lower bounds on the monotone
  complexity of some Boolean functions'', Soviet Math. Doklady,
  31:354-357, 1985.

\bibitem{Sm} Roman Smolensky: ``Algebraic Methods in the Theory of
  Lower Bounds for Boolean Circuit Complexity'', STOC 1987: 77-82.

\bibitem{Ta} Jun Tarui: ``Probablistic Polynomials, AC0 Functions, and
  the Polynomial-Time Hierarchy'', Theor. Comput. Sci. 113(1): 167-183
  (1993).

\bibitem{RY} Ran Raz, Amir Yehudayoff: ``Multilinear Formulas,
  Maximal-Partition Discrepancy and Mixed-Sources Extractors'', FOCS
  2008: 273-282.
\end{thebibliography}
\end{document}